\documentclass[a4paper,11pt]{article}
\usepackage[USenglish]{babel} 
\usepackage[T1]{fontenc}
\usepackage[ansinew]{inputenc}
\usepackage{lmodern}

\usepackage[margin=1in, paperwidth=8.4in, paperheight=11.5in]{geometry}

\usepackage{natbib}
\usepackage{fancyhdr}
\usepackage{lastpage}
\usepackage{textcomp}
\usepackage{xcolor}

\usepackage{amsmath}
\usepackage{amsthm}
\usepackage{amsfonts}
\usepackage{dsfont} 
\usepackage{stmaryrd} 

\usepackage{pst-all} 

\usepackage{graphicx}
\usepackage{multirow}
\usepackage[bookmarks = true, colorlinks = true]{hyperref}

\newtheorem{definition}{Definition}
\newtheorem{conclusion}{Conclusion}
\newtheorem{corollary}{Corollary}

\newtheorem{proposition}{Proposition}

\theoremstyle{definition}
\newtheorem{remark}{Remark}

\DeclareMathOperator*{\argmax}{arg\,max\,}
\DeclareMathOperator*{\argsup}{arg\,sup\,}
\DeclareMathOperator*{\arginf}{arg\,inf\,}

\begin{document}

\title{A Proximal Point Algorithm for Minimum Divergence Estimators with Application to Mixture Models}
\author{Diaa AL MOHAMAD\thanks{Corresponding author email: diaa.almohamad@gmail.com} \; \; \; Michel BRONIATOWSKI\thanks{email: michel.broniatowski@upmc.fr}\\
Laboratoire de Statistique Th\'eorique et Appliqu\'ee, Universit\'e Pierre et Marie Curie\\
4 place Jussieu 75005 PARIS}
\maketitle
\begin{abstract}
Estimators derived from a divergence criterion such as $\varphi-$divergences are generally more robust than the maximum likelihood ones. We are interested in particular in the so-called minimum dual $\phi-$divergence estimator (MD$\varphi$DE); an estimator built using a dual representation of $\varphi$--divergences. We present in this paper an iterative proximal point algorithm which permits to calculate such estimator. The algorithm contains by construction the well-known EM algorithm. Our work is based on the paper of \citep{Tseng} on the likelihood function. We provide some convergence properties by adapting the ideas of Tseng. We improve Tseng's results by relaxing the identifiability condition on the proximal term; a condition which is not verified for most mixture models and is hard to be verified for "non mixture" ones. Convergence of the EM algorithm in a two-component gaussian mixture is discussed in the spirit of our approach. Several experimental results on mixture models are provided to confirm the validity of the approach.
\end{abstract}
\section*{Introduction}
The EM algorithm is a well known method for calculating the maximum likelihood estimator of a model where incomplete data is considered. For example, when working with mixture models in the context of clustering, the labels or classes of observations are unknown during the training phase. Several variants of the EM algorithm were proposed, see \cite{McLachlanEM}. Another way to look at the EM algorithm is as a proximal point problem, see \cite{ChretienHero} and \cite{Tseng}. Indeed, one may rewrite the conditional expectation of the complete log-likelihood as a sum of the log-likelihood function and a distance-like function over the conditional densities of the labels provided an observation. Generally, the proximal term has a regularization effect in the sense that a proximal point algorithm is more stable and frequently outperforms classical optimization algorithms, see \cite{Goldstein}. Chr\'etien and Hero \cite{ChretienHeroAccel} prove superlinear convergence of a proximal point algorithm derived by the EM algorithm. Notice that EM-type algorithms usually enjoy no more than linear convergence.\\
Taking into consideration the need for robust estimators, and the fact that the MLE is the least robust estimator among the class of divergence-type estimators which we present below, we generalize the EM algorithm (and the version in \cite{Tseng}) by replacing the log-likelihood function by an estimator of a $\varphi-$divergence between the \emph{true distribution} of the data and the model. A
$\varphi-$divergence in the sense of Csisz\'{a}r \cite{Csiszar} is defined in the same way as \cite{BroniatowskiKeziou2007} by:
\[D_{\varphi}(Q,P) = \int{\varphi\left(\frac{dQ}{dP}(y)\right)dP(y)},\]
where $\varphi$ is a nonnegative strictly convex function. Examples of such divergences are: the Kullback-Leibler (KL) divergence for $\varphi(t)=t\log(t)-t+1$, the modified KL divergence for $\varphi(t)=-\log(t)+t-1$, the hellinger distance for $\varphi(t) = \frac{1}{2}(\sqrt{t}-1)$ among others. All these well-known divergences belong to the class of Cressie-Read functions \cite{CressieRead1984} defined by 
\begin{equation}
\varphi_{\gamma}(t) = \frac{x^{\gamma}-\gamma x + \gamma -1}{\gamma (\gamma -1)}  \text{ for } \gamma\in\mathbb{R}\setminus\{0,1\}.\label{eqn:CressieReadPhi}
\end{equation}
Since the $\varphi-$divergence calculus uses the unknown true distribution, we need to estimate it. We consider the dual estimator of the divergence introduced independently by \cite{BroniaKeziou2006} and \cite{LieseVajdaDivergence}. The use of this estimator is motivated by many reasons. Its minimum coincides with the MLE for $\varphi(t)=-\log(t)+t-1$. Besides, it has the same form for discrete and continuous models, and does not consider any partitioning or smoothing.\\
Let $(P_{\phi})_{\phi\in\Phi}$ be a parametric model with $\Phi\subset\mathbb{R}^d$, and denote $\phi^T$ the \emph{true} set of parameters. Let $dy$ be the Lebesgue measure defined on $\mathbb{R}$. Suppose that $\forall \phi\in\Phi$, the probability measure $P_{\phi}$ is absolutely continuous with respect to $dy$ and denote $p_{\phi}$ the corresponding probability density.
The dual estimator of the $\varphi-$divergence given an $n-$sample $y_1,\cdots,y_n$ is given by:
\begin{equation}
\hat{D}_{\varphi}(p_{\phi},p_{\phi_T}) = \sup_{\alpha\in\Phi}\int{\varphi'\left(\frac{p_{\phi}}{p_{\alpha}}\right)(x)p_{\phi}(x)dx} - \frac{1}{n}\sum_{i=1}^n{\varphi^{\#}\left(\frac{p_{\phi}}{p_{\alpha}}\right)(y_i)},
\label{eqn:DivergenceDef}
\end{equation}
with $\varphi^{\#}(t)=t\varphi'(t)-\varphi(t)$. AL Mohamad \cite{Diaa} argues that this formula works well under the model, however, when we are not, this quantity largely underestimates the divergence between the true distribution and the model, and proposes following modification:
\begin{equation}
\tilde{D}_{\varphi}(p_{\phi},p_{\phi_T}) = \int{\varphi'\left(\frac{p_{\phi}}{K_{n,w}}\right)(x)p_{\phi}(x)dx} - \frac{1}{n}\sum_{i=1}^n{\varphi^{\#}\left(\frac{p_{\phi}}{K_{n,w}}\right)(y_i)},
\label{eqn:NewDivergenceDef}
\end{equation}
where $K_{n,w}$ is the Rosenblatt-Parzen kernel estimate with window parameter $w$. Whether it is $\hat{D}_{\varphi}$, or $\tilde{D}_{\varphi}$, the minimum dual $\varphi-$divergence estimator (MD$\varphi$DE) is defined as the argument of the infimum of the dual approximation:
\begin{eqnarray}
\hat{\phi}_n & = & \arginf_{\phi\in\Phi} \hat{D}_{\varphi}(p_{\phi},p_{\phi_T}),
\label{eqn:MDphiDEClassique} \\
\tilde{\phi}_n & = & \arginf_{\phi\in\Phi} \tilde{D}_{\varphi}(p_{\phi},p_{\phi_T}).
\label{eqn:NewMDphiDE}
\end{eqnarray}
Asymptotic properties and consistency of these two estimators can be found in \cite{BroniatowskiKeziou2007} and \cite{Diaa}. Robustness properties were also studied using the influence function approach in \cite{TomaBronia} and \cite{Diaa}. The kernel-based MD$\varphi$DE (\ref{eqn:NewMDphiDE}) seems to be a \emph{better} estimator than the classical MD$\varphi$DE (\ref{eqn:MDphiDEClassique}) in the sense that the former is robust whereas the later is generally not. Under the model, the estimator given by (\ref{eqn:MDphiDEClassique}) is, however, more efficient especially when the true density of the data is unbounded\footnote{More investigation is needed here since we may use asymmetric kernels to overcome this difficulty.}, see \cite{Diaa} for a brief comparison. \\
Here in this paper, we propose to calculate the MD$\varphi$DE using an iterative procedure based on the work of \cite{Tseng} on the log-likelihood function. This procedure has the form of a proximal point algorithm, and extends the EM algorithm. Our convergence proof demands some regularity of the estimated divergence with respect to the parameter vector which is not simply checked using (\ref{eqn:DivergenceDef}). Recent results in the book of Rockafellar and Wets \cite{Rockafellar} provide sufficient conditions to solve this problem. Differentiability with respect to $\phi$ still remains a very hard task, therefore, our results cover cases when the objective function is not differentiable.\\
The paper is organized as follows: In Sect. 1, we present the general context. We also present the derivation of our algorithm from the EM algorithm and passing by Tseng's generalization. In Sect. 2, we present some convergence properties. We discuss in Sect. 3 a variant of the algorithm with a theoretical global infimum, and an exambple on the two-gaussian mixture model and a convergence proof of the EM algorithm in the spirit of our approach. Finally, Sect. 4 contains simulations confirming our claim about the efficiency and the robustness of our approach in comparison with the MLE. The algorithm is also applied to the so-called MDPD introduced by \cite{BasuMPD}.
\section{A Description of the Algorithm}
\subsection{General Context and Notations}
Let $(X,Y)$ be a couple of random variables with joint probability density function $f(x,y|\phi)$ parametrized by a vector of parameters $\phi\in\Phi\subset\mathbb{R}^d$. Let $(X_1,Y_1),\cdots,$ $(X_n,Y_n)$ be $n$ copies of $(X,Y)$ independently and identically distributed. Finally, let $(x_1,y_1),\cdots,(x_n,y_n)$ be $n$ realizations of the $n$ copies of $(X,Y)$. The $x_i$'s are the unobserved data (labels) and the $y_i$'s are the observations. The vector of parameters $\phi$ is unknown and need to be estimated. The observed data $y_i$ are supposed to be real numbers, and the labels $x_i$ belong to a space $\mathcal{X}$ not necessarily finite unless mentioned otherwise. The marginal density of the observed data is given by $p_{\phi}(y)=\int{f(x,y|\phi)}dx$, where $dx$ is a measure defined on the label space (for example, the counting measure if we work with mixture models).\\
For a parametrized function $f$ with a parameter $a$, we write $f(x|a)$. We use the notation $\phi^k$ for sequences with the index above. Derivatives of a real valued function $\psi$ defined on $\mathbb{R}$ are written as $\psi',\psi'',$ etc. We use $\nabla f$ for the gradient of a real function $f$ defined on $\mathbb{R}^d$, and $J_f$ for the matrix of second order partial derivatives. For a generic function of two (vectorial) arguments $D(\phi|\theta)$, then $\nabla_1 D(\phi|\theta)$ denotes the gradient with respect to the first (vectorial) variable. Finally, for any set $A$, we use $int(A)$ to denote the interior of $A$. 
\subsection{EM Algorithm and Tseng's Generalization}
\noindent The EM algorithm estimates the unknown parameter vector by (see \cite{Dempster}):
\[\phi^{k+1} = \argmax_{\Phi} \mathbb{E}\left[\log(f(\textbf{X},\textbf{Y}|\phi)) \left| \textbf{Y}=\textbf{y},\phi^k\right.\right].\]
where $\textbf{X} = (X_1,\cdots,X_n)$, $\textbf{Y} = (Y_1,\cdots,Y_n)$ and $\textbf{y}=(y_1,\cdots,y_n)$. By independence between the couples $(X_i,Y_i)$'s, previous iteration may be written as:
\begin{eqnarray}
\phi^{k+1} & = & \argmax_{\Phi} \sum_{i=1}^n{\mathbb{E}\left[\log(f(X_i,Y_i|\phi)) \left| Y_i=y_i,\phi^k\right.\right]} \nonumber\\
					 & = & \argmax_{\Phi} \sum_{i=1}^n\int_{\mathcal{X}}{\log(f(x,y_i|\phi)) h_i(x|\phi^k) dx},
\label{eqn:EMAlgo}
\end{eqnarray}
where $h_i(x|\phi^k)=\frac{f(x,y_i|\phi^k)}{p_{\phi^k}(y_i)}$ is the conditional density of the labels (at step $k$) provided $y_i$ which we suppose to be positive $dx-$almost everywhere. It is well-known that the EM iterations can be rewritten as a difference between the log-likelihood and a \emph{Kullback-Liebler} distance-like function. Indeed,
\begin{eqnarray*}
\phi^{k+1} & = & \argmax_{\Phi} \sum_{i=1}^n\int_{\mathcal{X}}{\log\left(h_i(x|\phi)\times p_{\phi}(y_i)\right) h_i(x|\phi^k) dx} \\
 					 & = & \argmax_{\Phi} \sum_{i=1}^n\int_{\mathcal{X}}{\log\left(p_{\phi}(y_i)\right) h_i(x|\phi^k) dx} + \sum_{i=1}^n\int_{\mathcal{X}}{\log\left(h_i(x|\phi)\right) h_i(x|\phi^k) dx} \\
           & = & \argmax_{\Phi} \sum_{i=1}^n{\log\left(p_{\phi}(y_i)\right)} + \sum_{i=1}^n\int_{\mathcal{X}}{\log\left(\frac{h_i(x|\phi)}{h_i(x|\phi^k)}\right) h_i(x|\phi^k) dx}\\
					& & \qquad \qquad \qquad \qquad + \sum_{i=1}^n\int_{\mathcal{X}}{\log\left(h_i(x|\phi^k)\right) h_i(x|\phi^k) dx}.
\end{eqnarray*}
The final line is justified by the fact that $h_i(x|\phi)$ is a density, therefore it integrates to 1. The additional term does not depend on $\phi$ and, hence, can be omitted. We now have the following iterative procedure:
\[\phi^{k+1} = \argmax_{\Phi} \sum_{i=1}^n{\log\left(g(y_i|\phi)\right)} + \sum_{i=1}^n\int_{\mathcal{X}}{\log\left(\frac{h_i(x|\phi)}{h_i(x|\phi^k)}\right) h_i(x|\phi^k) dx}.\]
The previous iteration has the form of a proximal point maximization of the log-likelihood, i.e. a perturbation of the log-likelihood by a distance-like function defined on the conditional densities of the labels. Tseng \cite{Tseng} generalizes this iteration by allowing any nonnegative convex function $\psi$ to replace the $t\mapsto-\log(t)$ function. Tseng's recurrence is defined by:
\begin{equation}
\phi^{k+1} = \argsup_{\phi} J(\phi) - D_{\psi}(\phi,\phi^k),
\label{eqn:TsengAlgo}
\end{equation}
where $J$ is the log-likelihood function and $D_{\psi}$ is given by:
\begin{equation}
D_{\psi}(\phi,\phi^k) = \sum_{i=1}^n\int_{\mathcal{X}}{\psi\left(\frac{h_i(x|\phi)}{h_i(x|\phi^k)}\right)h_i(x|\phi^k)dx},
\label{eqn:DivergenceClassesNtNorm}
\end{equation}
for any real nonnegative convex function $\psi$ such that $\psi(1)=\psi'(1)=0$. $D_{\psi}(\phi_1,\phi_2)$ is nonnegative, and $D_{\psi}(\phi_1,\phi_2)=0$ if and only if $\forall i, h_i(x|\phi_1) = h_i(x|\phi_2)$ $dx-$almost everywhere.\\

\subsection{Generalization of Tseng's Algorithm}\label{subsec:OurAlgo}
We use the relation between maximizing the log-likelihood and minimizing the Kullback-Liebler divergence to generalize the previous algorithm. We, therefore, replace the log-likelihood function by an estimate of a $\varphi-$divergence $D_{\varphi}$ between the true distribution and the model. We use the dual estimator of the divergence presented earlier in the introduction (\ref{eqn:DivergenceDef}) or (\ref{eqn:NewDivergenceDef}) which we denote in the same manner $\hat{D}_{\varphi}$ unless mentioned otherwise. Our new algorithm is defined by:
\begin{equation}
\phi^{k+1} = \arginf_{\phi} \hat{D}_{\varphi}(p_{\phi},p_{\phi_T}) + \frac{1}{n}D_{\psi}(\phi,\phi^k),
\label{eqn:DivergenceAlgoPreVersion}
\end{equation}
where $D_{\psi}(\phi,\phi^k)$ is defined by (\ref{eqn:DivergenceClassesNtNorm}). When $\varphi(t) = -\log(t)+t-1$, it is easy to see that we get recurrence (\ref{eqn:TsengAlgo}). Indeed, for the case of (\ref{eqn:DivergenceDef}) we have:
\[\hat{D}_{\varphi}(p_{\phi},p_{\phi_T}) = \sup_{\alpha} \frac{1}{n}\sum_{i=1}^n{\log(p_{\alpha}(y_i))} - \frac{1}{n}\sum_{i=1}^n{\log(p_{\phi}(y_i))}.\]
Using the fact that the first term in $\hat{D}_{\varphi}(p_{\phi},p_{\phi_T})$ does not depend on $\phi$, so it does not count in the $\arginf$ defining $\phi^{k+1}$, we easily get (\ref{eqn:TsengAlgo}). The same applies for the case of (\ref{eqn:NewDivergenceDef}). For notational simplicity, from now on, we redefine $D_{\psi}$ with a normalization by $n$, i.e. 
\begin{equation}
D_{\psi}(\phi,\phi^k) = \frac{1}{n} \sum_{i=1}^n\int_{\mathcal{X}}{\psi\left(\frac{h_i(x|\phi)}{h_i(x|\phi^k)}\right)h_i(x|\phi^k)dx}.
\label{eqn:DivergenceClasses}
\end{equation}
Hence, our set of algorithms is redefined by:
\begin{equation}
\phi^{k+1} = \arginf_{\phi} \hat{D}_{\varphi}(p_{\phi},p_{\phi_T}) + D_{\psi}(\phi,\phi^k).
\label{eqn:DivergenceAlgo}
\end{equation}
We will see later that this iteration forces the divergence to decrease and that under suitable conditions, it converges to a (local) minimum of $\hat{D}_{\varphi}(p_{\phi},p_{\phi_T})$. It results that, algorithm (\ref{eqn:DivergenceAlgo}) is a way to calculate both the MD$\varphi$DE (\ref{eqn:MDphiDEClassique}) and the kernel-based MD$\varphi$DE (\ref{eqn:NewMDphiDE}). 
\section{Some Convergence Properties of \texorpdfstring{$\phi^k$}{phi k}}\label{sec:Proofs}
We show here how, according to some possible situations, one may prove convergence of the algorithm defined by (\ref{eqn:DivergenceAlgo}). Let $\phi^0$ be a given initialization, and let 
\[\Phi^0 := \{\phi\in\Phi: \hat{D}_{\varphi}(p_{\phi},p_{\phi_T})\leq \hat{D}_{\varphi}(p_{\phi^0},p_{\phi_T})\}\] 
which we suppose to be a subset of $int(\Phi)$. The idea of defining such a set in this context is inherited from the paper of \citep{Wu} which provided the first \emph{correct proof} of convergence for the EM algorithm. Before going any further, we recall the following Definition of a (generalized) stationary point.\\ 
\begin{definition}
Let $f:\mathbb{R}^d\rightarrow\mathbb{R}$ be a real valued function. If $f$ is differentiable at a point $\phi^*$ such that $\nabla f(\phi^*)=0$, we then say that $\phi^*$ is a stationary point of $f$. If $f$ is not differentiable at $\phi^*$ but the subgradient of $f$at $\phi^*$, say $\partial f(\phi^*)$, exists such that $0\in\partial f(\phi^*)$, then $\phi^*$ is called a generalized stationary point of $f$.
\end{definition}
We will be using the following assumptions:
\begin{itemize}
\item[A0.] Functions $\phi\mapsto\hat{D}_{\varphi}(p_{\phi}|p_{\phi_T}), D_{\psi}$ are lower semicontinuous.
\item[A1.] Functions $\phi\mapsto\hat{D}_{\varphi}(p_{\phi}|p_{\phi_T}), D_{\psi}$ and $\nabla_1 D_{\psi}$ are defined and continuous on, respectively, $\Phi, \Phi\times\Phi$ and $\Phi\times\Phi$;
\item[AC.] $\phi\mapsto\nabla \hat{D}_{\varphi}(p_{\phi}|p_{\phi_T})$ is defined and continuous on $\Phi$
\item[A2.] $\Phi^0$ is a compact subset of int$(\Phi)$;
\item[A3.] $D_{\psi}(\phi,\bar{\phi})>0$ for all $\bar{\phi}\neq \phi \in \Phi$.
\end{itemize}
Recall also that we suppose that $h_i(x|\phi)>0, dx-a.e.$ We relax the convexity assumption of function $\psi$. We only suppose that $\psi$ is nonnegative and $\psi(t)=0$ iff $t=1$. Besides $\psi'(t)=0$ if $t = 1$.\\
Continuity and differentiability assumptions of function $\phi\mapsto \hat{D}_{\varphi}(p_{\phi}|p_{\phi_T})$ for the case of (\ref{eqn:NewDivergenceDef}) can be easily checked using Lebesgue theorems. Continuity assumption for the case of (\ref{eqn:DivergenceDef}) can be checked using Theorem 1.17 or Corollary 10.14 in \cite{Rockafellar}. Differentiability can also be checked using Corollary 10.14 or Theorem 10.31 in the same book. In what concerns $D_{\psi}$, continuity and differentiability can be obtained merely by fulfilling Lebesgue theorems conditions. When working with mixture models, we only need the continuity and differentiability of $\psi$ and functions $h_i$. The later is easily deduced from regularity assumptions on the model. For assumption A2, there is no universal method, see paragraph (\ref{Example:GaussMix}) for an Example. Assumption A3 can be checked using Lemma 2 in \cite{Tseng}.\\
We start the convergence properties by proving that the objective function $\hat{D}_{\varphi}(p_{\phi}|p_{\phi_T})$ decreases alongside the the sequence $(\phi^k)_k$, and a possible set of conditions for the existence of the sequence $(\phi^k)_k$. \\
\begin{proposition} \label{prop:DecreaseDphi}
 (a)Assume that the sequence $(\phi^k)_k$ is well defined in $\Phi$, then $\hat{D}_{\varphi}(p_{\phi^{k+1}}|p_{\phi_T})\leq \hat{D}_{\varphi}(p_{\phi^k}|p_{\phi_T})$, and (b) $\forall k, \phi^k \in \Phi^0$. (c) Assume A0 and A2 are verified, then the sequence $(\phi^k)_k$ is defined and bounded. Moreover, the sequence $(\hat{D}_{\varphi}(p_{\phi^k}|p_{\phi^T}))_k$ converges.
\end{proposition}
\begin{proof}
\underline{We prove $(a)$}. we have by Definition of the arginf:
\[\hat{D}_{\varphi}(p_{\phi^{k+1}},p_{\phi_T}) + D_{\psi}(\phi^{k+1},\phi^k) \leq \hat{D}_{\varphi}(p_{\phi^k},p_{\phi_T}) + D_{\psi}(\phi^k,\phi^k).\]
We use the fact that $D_{\psi}(\phi^k,\phi^k)=0$ for the right hand and that $D_{\psi}(\phi^{k+1},\phi^k)\geq 0$ for the left hand side of the previous inequality. Hence $\hat{D}_{\varphi}(p_{\phi^{k+1}},p_{\phi_T})\leq \hat{D}_{\varphi}(p_{\phi^k},p_{\phi_T})$.\\
\underline{We prove $(b)$}. Using the decreasing property previously proved in (a), we have by recurrence $\forall k, \hat{D}_{\varphi}(p_{\phi^{k+1}},p_{\phi_T})\leq \hat{D}_{\varphi}(p_{\phi^k},p_{\phi_T})\leq\cdots \leq \hat{D}_{\varphi}(p_{\phi^0},p_{\phi_T})$. The result follows for both algorithms directly by Definition of $\Phi^0$.\\
\underline{We prove $(c)$}. By induction on $k$. For $k=0$, clearly $\phi^0 = (\lambda^0,\theta^0)$ is well defined (a choice we make\footnote{The choice of the initial point of the sequence may influence the convergence of the sequence. See the Example of the Gaussian mixture in paragraph (\ref{Example:GaussMix}).}). Suppose for some $k\geq 0$ that $\phi^k = (\lambda^k,\theta^k)$ exists. We prove that the infimum is attained in $\Phi^0$. Let $\phi\in\Phi$ be any vector at which the value of the optimized function has a value less than its value at $\phi^k$, i.e. $\hat{D}_{\varphi}(p_{\phi},p_{\phi_T}) + D_{\psi}(\phi,\phi^k)\leq \hat{D}_{\varphi}(p_{\phi^k},p_{\phi_T}) + D_{\psi}(\phi^k,\phi^k)$. We have:
\begin{eqnarray*}
\hat{D}_{\varphi}(p_{\phi},p_{\phi_T}) & \leq & \hat{D}_{\varphi}(p_{\phi},p_{\phi_T}) + D_{\psi}(\phi,\phi^k) \\ 
& \leq & \hat{D}_{\varphi}(p_{\phi^k},p_{\phi_T}) + D_{\psi}(\phi^k,\phi^k) \\
  & \leq & \hat{D}_{\varphi}(p_{\phi^k},p_{\phi_T}) \\
  & \leq & \hat{D}_{\varphi}(p_{\phi^0},p_{\phi_T}).
\end{eqnarray*}
The first line follows from the non negativity of $D_{\psi}$. As $\hat{D}_{\varphi}(p_{\phi},p_{\phi_T})\leq \hat{D}_{\varphi}(p_{\phi^0},p_{\phi_T})$, then $\phi\in\Phi^0$. Thus, the infimum can be calculated for vectors in $\Phi^0$ instead of $\Phi$. Since $\Phi^0$ is compact and the optimized function is lower semicontinuous (the sum of two lower semicontinuous functions), then  the infimum exists and is attained in $\Phi^0$. We may now define $\phi^{k+1}$ to be a vector whose corresponding value is equal to the infimum.\\
Convergence of the sequence $(\hat{D}_{\varphi}(p_{\phi^k},p_{\phi_T}))_k$ comes from the fact that it is non increasing and bounded. It is non increasing by virtue of (a). Boundedness comes from the lower semicontinuity of $\phi\mapsto\hat{D}_{\varphi}(p_{\phi},p_{\phi_T})$. Indeed, $\forall k, \hat{D}_{\varphi}(p_{\phi^k},p_{\phi_T}) \geq \inf_{\phi\in\Phi^0}\hat{D}_{\varphi}(p_{\phi},p_{\phi_T})$. The infimum of a proper lower semicontinuous function on a compact set exists and is attained on this set. Hence, the quantity $\inf_{\phi\in\Phi^0}\hat{D}_{\varphi}(p_{\phi},p_{\phi_T})$ exists and is finite. This ends the proof.
\end{proof}
Compactness in part (c) can be replaced by inf-compactness of function $\phi\mapsto \hat{D}_{\varphi}(p_{\phi}|p_{\phi_T})$ and continuity of $D_{\psi}$ with respect to its first argument. The convergence of the sequence $(\hat{D}_{\varphi}(\phi^k|\phi_T))_k$ is an interesting property, since in general there is no theoretical guarantee, or it is difficult to prove that the whole sequence $(\phi^k)_k$ converges. It may also continue to fluctuate around a minimum. The decrease of the error criterion $\hat{D}_{\varphi}(\phi^k|\phi_T)$ between two iterations helps us decide when to stop the iterative procedure.
\begin{proposition} \label{prop:StationaryPhiDiff}
 Suppose A1 verified, $\Phi^0$ is closed and $\{\phi^{k+1}-\phi^k\}\rightarrow 0$.
\begin{itemize}
\item[(a)] If AC is verified, then any limit point of $(\phi^k)_k$ is a stationary point of $\phi\mapsto\hat{D}_{\varphi}(p_{\phi}|p_{\phi^T})$;
\item[(b)] If AC is dropped,  then any limit point of $(\phi^k)_k$ is a "generalized" stationary point of $\phi\mapsto\hat{D}_{\varphi}(p_{\phi}|p_{\phi^T})$, i.e. zero belongs to the subgradient of $\phi\mapsto\hat{D}_{\varphi}(p_{\phi}|p_{\phi^T})$ calculated at the limit point.
\end{itemize}
\end{proposition}
\begin{proof} 
\underline{We prove $(a)$}. Let $(\phi^{n_k})_k$ be a convergent subsequence of $(\phi^k)_k$ which converges to $\phi^{\infty}$. First, $\phi^{\infty} \in \Phi^0$, because $\Phi^0$ is closed and the subsequence $(\phi^{n_k})$ is a sequence of elements of $\Phi^0$ (proved in Proposition \ref{prop:DecreaseDphi}.b).\\
Let's show now that the subsequence $(\phi^{n_k+1})$ also converges to $\phi^{\infty}$. We simply have:
\begin{eqnarray*}
\|\phi^{n_k+1} - \phi^{\infty}\| & \leq & \|\phi^{n_k} - \phi^{\infty}\| + \|\phi^{n_k+1} - \phi^{n_k}\|
\end{eqnarray*}
Since $\phi^{k+1} - \phi^k \rightarrow 0$ and $\phi^{n_k} \rightarrow \phi^{\infty}$, we conclude that $\phi^{n_k+1} \rightarrow \phi^{\infty}$.\\
By Definition of $\phi^{n_k+1}$, it verifies the infimum in recurrence (\ref{eqn:DivergenceAlgo}), so that the gradient of the optimized function is zero:
\[\nabla \hat{D}_{\varphi}(p_{\phi^{n_k+1}},p_{\phi_T}) + \nabla D_{\psi}(\phi^{n_k+1},\phi^{n_k}) = 0\]
Using the continuity assumptions A1 and AC of the gradients, one can pass to the limit with no problem:
\[\nabla \hat{D}_{\varphi}(p_{\phi^{\infty}},p_{\phi_T}) + \nabla D_{\psi}(\phi^{\infty},\phi^{\infty}) = 0\]
However, the gradient $\nabla D_{\psi}(\phi^{\infty},\phi^{\infty})=0$ because (recall that $\psi'(1)=0$) for any $\phi\in\Phi$
\[\nabla D_{\psi}(\phi,\phi) = \sum_{i=1}^n\int_{\mathcal{X}}{\frac{\nabla h_i(x|\phi)}{h_i(x|\phi)}\psi'\left(\frac{h_i(x|\phi)}{h_i(x|\phi)}\right)h_i(x|\phi)dx}=\sum_{i=1}^n\int_{\mathcal{X}}{\nabla h_i(x|\phi)\psi'(1)dx}\]
which is equal to zero since $\psi'(1)=0$. This implies that $\nabla \hat{D}_{\varphi}(p_{\phi^{\infty}},p_{\phi_T})=0$.\\
\underline{We prove (b)}. We use again the Definition of the arginf. As the optimized function is not necessarily differentiable at the points of the sequence $\phi^k$, a necessary condition for $\phi^{k+1}$ to be an infimum is that 0 belongs to the subgradient of the function on $\phi^{k+1}$. Since $D_{\psi}(\phi,\phi^k)$ is assumed to be differentiable, the optimality condition is translated into:
\[-\nabla D_{\psi}(\phi^{k+1},\phi^k) \in \partial \hat{D}_{\varphi}(p_{\phi^{k+1}},p_{\phi_T})\quad \forall k\]
Since $\hat{D}_{\varphi}(p_{\phi},p_{\phi_T})$ is continuous, then its subgradient is outer semicontinuous (see \citep{Rockafellar} Chap 8, Proposition 7). We use the same arguments presented in (a) to conclude the existence of two subsequences $(\phi^{n_k})_k$ and $(\phi^{n_k+1})_k$ which converge to the same limit $\phi^{\infty}$. By Definition of outer semicontinuity, and since $\phi^{n_k+1}\rightarrow\phi^{\infty}$, we have:
\begin{equation}
\limsup_{\phi^{n_k+1}\rightarrow\phi^{\infty}} \partial \hat{D}_{\varphi}(p_{\phi^{n_k+1}},p_{\phi_T})\subset \partial \hat{D}_{\varphi}(p_{\phi^{\infty}},p_{\phi_T})
\label{eqn:subgradInc}
\end{equation}
We want to prove that $0\in\limsup_{\phi^{n_k+1}\rightarrow\phi^{\infty}} \partial \hat{D}_{\varphi}(p_{\phi^{n_k+1}},p_{\phi_T})$. By Definition of limsup\footnote{We use here the Definition corresponding to the outer limit, see \citep{Rockafellar} Chap 4, Definition 1 or Chap 5-B.}:
\[\limsup_{\phi\rightarrow\phi^{\infty}} \partial \hat{D}_{\varphi}(p_{\phi},p_{\phi_T}) = \left\{ u|\exists \phi^k\rightarrow\phi^{\infty},\exists u^k\rightarrow u \text{ with } u^k\in \partial \hat{D}_{\varphi}(p_{\phi^k},p_{\phi_T})\right\}.\]
In our scenario, $\phi = \phi^{n_k+1}$, $\phi^k = \phi^{n_k+1}$, $u = 0$ and $u^k = \nabla_1 D_{\psi}(\phi^{n_k+1},\phi^{n_k})$. The continuity of $\nabla_1 D_{\psi}$ with respect to both arguments and the fact that the two subsequences $\phi^{n_k+1}$ and $\phi^{n_k}$ converge to the same limit, imply that $u^k\rightarrow\nabla_1 D_{\psi}(\phi^{\infty},\phi^{\infty}) = 0$. Hence $u = 0\in\limsup_{\phi^{n_k+1}\rightarrow\phi^{\infty}} \partial \hat{D}_{\varphi}(p_{\phi^{n_k+1}},p_{\phi_T})$. By inclusion (\ref{eqn:subgradInc}), we get our result:
\[0\in \partial \hat{D}_{\varphi}(p_{\phi^{\infty}},p_{\phi_T}),\]
This ends the proof.
\end{proof}
\noindent The assumption $\{\phi^{k+1}-\phi^k\}\rightarrow 0$ used in Proposition \ref{prop:StationaryPhiDiff} is not easy to be checked unless one has a close formula of $\phi^k$. The following Proposition gives a method to prove such assumption. This method seems simpler, but it is not verified in many mixture models, see Sect. (\ref{Example:GaussMix}) for a counter Example.\\
\begin{proposition} \label{prop:PhiDiffConverge}
Assume that A1, A2 and A3 are verified, then $\{\phi^{k+1}-\phi^k\}\rightarrow 0$. Thus, by Proposition \ref{prop:StationaryPhiDiff} (according to whether AC is verified or not) any limit point of the sequence $\phi^k$ is a (generalized) stationary point of $\hat{D}_{\varphi}(.|\phi_T)$.
\end{proposition}
\begin{proof}
By contradiction, let's suppose that $\phi^{k+1}-\phi^k$ does not converge to 0. There exists a subsequence such that $\|\phi^{N_0(k)+1}-\phi^{N_0(k)}\| > \varepsilon,\; \forall k\geq k_0$. Since $(\phi^k)_k$ belongs to the compact set $\Phi^0$, there exists a convergent subsequence $(\phi^{N_1\circ N_0(k)})_k$ such that $\phi^{N_1\circ N_0(k)}\rightarrow \bar{\phi}$. The sequence $(\phi^{N_1\circ N_0(k)+1})_k$ belongs to the compact set $\Phi^0$, therefore we can extract a further subsequence $(\phi^{N_2\circ N_1\circ N_0(k)+1})_k$ such that $\phi^{N_2\circ N_1\circ N_0(k)+1}\rightarrow \tilde{\phi}$. Besides $\hat{\phi}\neq \tilde{\phi}$. Finally since the sequence $(\phi^{N_1\circ N_0(k)})_k$ is convergent, a further subsequence also converges to the same limit $\bar{\phi}$. We have proved the existence of a subsequence of $(\phi^k)_k$ such that $\phi^{N(k)+1}-\phi^{N(k)}$ does not converge to 0 and such that $\phi^{N(k)+1} \rightarrow \tilde{\phi}$, $\phi^{N(k)} \rightarrow \bar{\phi}$ with $\bar{\phi} \neq \tilde{\phi}$.\\
The real sequence $\hat{D}_{\varphi}(p_{\phi^k},p_{\phi_T})_k$ converges as proved in Proposition \ref{prop:DecreaseDphi}-c. As a result, both sequences $\hat{D}_{\varphi}(p_{\phi^{N(k)+1}},p_{\phi_T})$ and $\hat{D}_{\varphi}(p_{\phi^{N(k)}},p_{\phi_T})$ converge to the same limit being subsequences of the same convergent sequence. In the proof of Proposition \ref{prop:DecreaseDphi}, we can deduce the following inequality:
\begin{equation}
\hat{D}(p_{\lambda^{k+1},\theta^{k+1}},p_{\phi_T}) + D_{\psi}((\lambda^{k+1},\theta^{k+1}),\phi^k) \leq \hat{D}(p_{\lambda^k,\theta^k},p_{\phi_T})
\label{eqn:DivergenceDecreaseSeq}
\end{equation}
which is also verified to any substitution of $k$ by $N(k)$. By passing to the limit on k, we get $D_{\psi}(\tilde{\phi},\bar{\phi}) \leq 0$. However, the distance-like function $D_{\psi}$ is positive, so that it becomes zero. Using assumption A3, $D_{\psi}(\tilde{\phi},\bar{\phi}) = 0$ implies that $\tilde{\phi} = \bar{\phi}$. This contradicts the hypothesis that $\phi^{k+1}-\phi^k$ does not converge to 0.\\
The second part of the Proposition is a direct result of Proposition \ref{prop:StationaryPhiDiff}.
\end{proof}
\begin{corollary} \label{Cor:TotalConverg}
Under assumptions of Proposition \ref{prop:PhiDiffConverge}, the set of accumulation points of $(\phi^k)_k$ is a connected compact set. Moreover, if $\phi\mapsto\hat{D}(p_{\phi},p_{\phi_T})$ is strictly convex in a neighborhood of a limit point of the sequence $(\phi^k)_k$, then the whole sequence $(\phi^k)_k$ converges to a local minimum of $\hat{D}(p_{\phi},p_{\phi_T})$.
\end{corollary}
\begin{proof}
Since the sequence $(\phi)_k$ is bounded and verifies $\phi^{k+1}-\phi^k\rightarrow 0$, then Theorem 28.1 in \citep{Ostrowski} implies that the set of accumulation points of $(\phi^k)_k$ is a connected compact set. It is not empty since $\Phi^0$ is compact. Let $\phi^{\infty}$ be a limit point of $(\phi^k)_k$. The assumption about strict convexity of $\hat{D}(p_{\phi},p_{\phi_T})$ in a neighborhood of $\phi^{\infty}$ implies that it is isolated in the sense that if there are another limit point $\tilde{\phi}$, then there is $\varepsilon>0$ such that $\|\phi^{\infty} - \tilde{\phi}\|>\varepsilon$. Hence, the set of accumulation points can be written as the union of at least two disjoint open sets which contradicts the connectedness property. Thus, $\phi^{\infty}$ is the only limit point of the sequence $(\phi^k)$. To end the proof, we need to show that the whole sequence converge. By contradiction, if it does not converge, there exists then $\varepsilon>0$ and an infinity of terms which verifies $\|\phi^{N_0(k)} - \phi^{\infty}\|>\varepsilon$. By compactness of $\Phi^0$, one may extract a subsequence of $(\phi^{N_0(k)})_k$, say $(\phi^{N_1\circ N_0(k)})_k$, which converges to some $\hat{\phi}$. Moreover, by continuity of the euclidean norm, $\|\phi^{N_1\circ N_0(k)} - \phi^{\infty}\|\rightarrow \|\hat{\phi} - \phi^{\infty}\|$. Hence $\|\hat{\phi} - \phi^{\infty}\|\geq\varepsilon$. Contradiction is reached by uniqueness of the limit point of the sequence $(\phi^k)_k$. 
\end{proof}
Proposition \ref{prop:PhiDiffConverge} and Corollary \ref{Cor:TotalConverg} describe what we may hope to get of the sequence $\phi^k$. Convergence of the whole sequence is bound by a local convexity assumption in a neighborhood of a limit point. Although simple, this assumption remains difficult to be checked since we do not know where might be the limit points. Besides, assumption A3 is very restrictive, and is not verified in mixture models.\\
Proposition \ref{prop:StationaryPhiDiff} and \ref{prop:PhiDiffConverge} were developed for the likelihood function in the paper of \cite{Tseng}. Similar results for a general class of functions replacing $\hat{D}_{\varphi}$ and $D_{\psi}$ which may not be differentiable (but still continuous) are presented in \cite{ChretienHero}.  In these results, assumption A3 is essential. Although \cite{ChretienHeroProxGener} overcomes this problem, their approach demands that the log-likelihood has $-\infty$ limit as $\|\phi\|\rightarrow\infty$. This is simply not verified for mixture models. We present a similar method to \cite{ChretienHeroProxGener} based on the idea of \cite{Tseng} of using the set $\Phi^0$ which is valid for mixtures. We lose, however, the guarantee of consecutive decrease of the sequence.\\
\begin{proposition} \label{prop:NewRes}
Assume A1, AC and A2 verified. Any limit point of the sequence $(\phi^k)_k$ is a stationary point of $\phi\rightarrow\hat{D}(p_{\phi},p_{\phi_T})$. If AC is dropped, then 0  belongs to the subgradient of $\phi\mapsto\hat{D}(p_{\phi},p_{\phi_T})$ calculated at the limit point.
\end{proposition}
\begin{proof} 
If $(\phi^k)_k$ converges to, say, $\phi^{\infty}$, then the result falls simply from Proposition \ref{prop:StationaryPhiDiff}.\\
If $(\phi^k)_k$ does not converge. Since $\Phi^0$ is compact and $\forall k, \phi^k\in\Phi^0$ (proved in Proposition \ref{prop:DecreaseDphi}), there exists a subsequence $(\phi^{N_0(k)})_k$ such that $\phi^{N_0(k)}\rightarrow\tilde{\phi}$. Let's take the subsequence $(\phi^{N_0(k)-1})_k$. This subsequence does not necessarily converge; still it is contained in the compact $\Phi^0$, so that we can extract a further subsequence $(\phi^{N_1\circ N_0(k)-1})_k$ which converges to, say, $\bar{\phi}$. Now, the subsequence $(\phi^{N_1\circ N_0(k)})_k$ converges to $\tilde{\phi}$, because it is a subsequence of $(\phi^{N_0(k)})_k$. We have proved until now the existence of two convergent subsequences $\phi^{N(k)-1}$ and $\phi^{N(k)}$ with \emph{a priori} different limits. For simplicity and without any loss of generality, we will consider these subsequences to be $\phi^k$ and $\phi^{k+1}$ respectively.\\
Conserving previous notations, suppose that $\phi^{k+1}\rightarrow \tilde{\phi}$ and $\phi^{k}\rightarrow \bar{\phi}$. We use again inequality (\ref{eqn:DivergenceDecreaseSeq}):
\[\hat{D}(p_{\phi^{k+1}},p_{\phi_T}) + D_{\psi}(\phi^{k+1},\phi^k) \leq \hat{D}(p_{\lambda^k,\theta^k},p_{\phi_T})\]
By taking the limits of the two parts of the inequality as $k$ tends to infinity, and using the continuity of the two functions, we have 
\[\hat{D}(p_{\tilde{\phi}},p_{\phi_T}) + D_{\psi}(\tilde{\phi},\bar{\phi}) \leq \hat{D}(p_{\bar{\phi}},p_{\phi_T})\]
Recall that under A1-2, the sequence $\left(\hat{D}_{\varphi}(p_{\phi^k},p_{\phi_T})\right)_k$ converges, so that it has the same limit for any subsequence, i.e. $\hat{D}(p_{\tilde{\phi}},p_{\phi_T}) = \hat{D}(p_{\bar{\phi}},p_{\phi_T})$. We also use the fact that the distance-like function $D_{\psi}$ is non negative to deduce that $D_{\psi}(\tilde{\phi},\bar{\phi}) = 0$. Looking closely at the Definition of this divergence (\ref{eqn:DivergenceClasses}), we get that if the sum is zero, then each term is also zero since all terms are non negative. This means that:
\[\forall i\in\{1,\cdots,n\}, \quad \int_{\mathcal{X}}{\psi\left(\frac{h_i(x|\tilde{\phi})}{h_i(x|\bar{\phi})}\right)h_i(x|\bar{\phi})dx} = 0\]
The integrands are non negative functions, so they vanish almost ever where with respect to the measure $dx$ defined on the space of labels.
\[\forall i\in\{1,\cdots,n\}, \quad \psi\left(\frac{h_i(x|\tilde{\phi})}{h_i(x|\bar{\phi})}\right)h_i(x|\bar{\phi}) = 0\quad dx-a.e.\]
The conditional densities $h_i$ are supposed to be positive\footnote{In the case of two Gaussian (or more generally exponential) components, this is justified by virtue of a suitable choice of the initial condition.}, i.e. $ h_i(x|\bar{\phi})>0, dx-a.e.$. Hence, $\psi\left(\frac{h_i(x|\tilde{\phi})}{h_i(x|\bar{\phi})}\right) = 0, dx-a.e.$. On the other hand, $\psi$ is chosen in a way that $\psi(z)=0$ iff $z=1$, therefore :
\begin{equation}
\forall i\in\{1,\cdots,n\},\quad h_i(x|\tilde{\phi}) = h_i(x|\bar{\phi}) \quad dx-a.e.
\label{eqn:ProportionsEquality}
\end{equation}
Since $\phi^{k+1}$ is, by Definition, an infimum of $\phi\mapsto\hat{D}(p_{\phi},p_{\phi_T}) + D_{\psi}(\phi,\phi^k)$, then the gradient of this function is zero on $\phi^{k+1}$. It results that:
\[\nabla \hat{D}(p_{\phi^{k+1}},p_{\phi_T}) + \nabla D_{\psi}(\phi^{k+1},\phi^k) = 0,\quad \forall k\]
Taking the limit on $k$, and using the continuity of the derivatives, we get that:
\begin{equation}
\nabla \hat{D}(p_{\tilde{\phi}},p_{\phi_T}) + \nabla D_{\psi}(\tilde{\phi},\bar{\phi}) = 0
\label{eqn:GradientLimit}
\end{equation}
Let's write explicitly the gradient of the second divergence:
\[\nabla D_{\psi}(\tilde{\phi},\bar{\phi}) = \sum_{i=1}^n\int_{\mathcal{X}}{\frac{\nabla h_i(x|\tilde{\phi})}{h_i(x|\bar{\phi})}\psi'\left(\frac{h_i(x|\tilde{\phi})}{h_i(x|\bar{\phi})}\right)h_i(x|\bar{\phi})}\]
We use now the identities (\ref{eqn:ProportionsEquality}), and the fact that $\psi'(1)=0$, to deduce that:
\[\nabla D_{\psi}(\tilde{\phi},\bar{\phi}) = 0 \]
This entails using (\ref{eqn:GradientLimit}) that $\nabla \hat{D}(p_{\tilde{\phi}},p_{\phi_T}) = 0$.\\ 
Comparing the proved result with the notation considered at the beginning of the proof, we have proved that the limit of the subsequence $(\phi^{N_1\circ N_0(k)})_k$ is a stationary point of the objective function. Therefore, The final step is to deduce the same result on the original convergent subsequence $(\phi^{N_0(k)})_k$. This is simply due to the fact that $(\phi^{N_1\circ N_0(k)})_k$ is a subsequence of the convergent sequence $(\phi^{N_0(k)})_k$, hence they have the same limit. \\
\textbf{When assumption AC is dropped,} similar arguments to those used in the proof of Proposition \ref{prop:StationaryPhiDiff}-b. are employed. The optimality condition in (\ref{eqn:DivergenceAlgo}) implies :
\[-\nabla D_{\psi}(\phi^{k+1},\phi^k) \in \partial \hat{D}_{\varphi}(p_{\phi^{k+1}},p_{\phi_T})\quad \forall k\]
Function $\phi\mapsto\hat{D}_{\varphi}(p_{\phi},p_{\phi_T})$ is continuous, hence its subgradient is outer semicontinuous and:
\begin{equation}
\limsup_{\phi^{k+1}\rightarrow\phi^{\infty}} \partial \hat{D}_{\varphi}(p_{\phi^{k+1}},p_{\phi_T})\subset \partial \hat{D}_{\varphi}(p_{\tilde{\phi}},p_{\phi_T})
\label{eqn:OSCInclusion}
\end{equation}
By Definition of limsup:
\[\limsup_{\phi\rightarrow\phi^{\infty}} \partial \hat{D}_{\varphi}(p_{\phi},p_{\phi_T}) = \left\{ u|\exists \phi^k\rightarrow\phi^{\infty},\exists u^k\rightarrow u \text{ with } u^k\in \partial \hat{D}_{\varphi}(p_{\phi^k},p_{\phi_T})\right\}\]
In our scenario, $\phi = \phi^{k+1}$, $\phi^k = \phi^{k+1}$, $u = 0$ and $u^k = \nabla_1 D_{\psi}(\phi^{k+1},\phi^{k})$. We have proved above in this proof that $\nabla_1 D_{\psi}(\tilde{\phi},\bar{\phi}) = 0$ using only the convergence of $(\hat{D}_{\varphi}(p_{\phi^k},p_{\phi_T}))_k$, inequality (\ref{eqn:DivergenceDecreaseSeq}) and the properties of $D_{\psi}$. Assumption AC was not needed. Hence, $u^k\rightarrow 0$. This proves that, $u = 0\in\limsup_{\phi^{k+1}\rightarrow\phi^{\infty}} \partial \hat{D}_{\varphi}(p_{\phi^{n_k+1}},p_{\phi_T})$. Finally, using the inclusion (\ref{eqn:OSCInclusion}), we get our result:
\[0\in \partial \hat{D}_{\varphi}(p_{\tilde{\phi}},p_{\phi_T}),\]
which ends the proof.
\end{proof}

The proof of the previous proposition is very similar to the proof of Proposition \ref{prop:StationaryPhiDiff}. The key idea is to use the sequence of conditional densities $h_i(x|\phi^k)$ instead of the sequence $\phi^k$. According to the application, one may be interested only in Proposition \ref{prop:DecreaseDphi} or in Propositions \ref{prop:StationaryPhiDiff}-\ref{prop:NewRes}. If one is interested in the parameters, Propositions \ref{prop:StationaryPhiDiff} to \ref{prop:NewRes} should be used, since we need a stable limit of $(\phi^k)_k$. If we are only interested in minimizing an error criterion $\hat{D}_{\varphi}(p_{\phi},p_{\phi^T})$ between the estimated distribution and the true one, Proposition \ref{prop:DecreaseDphi} should be sufficient.
\section{Case studies}
\subsection{An algorithm with theoretically global infimum attainment}\label{subsec:GlobaInf}
We present a variant of algorithm (\ref{eqn:DivergenceAlgo}) which ensures \emph{theoretically} the convergence to a global infimum of the objective function $\hat{D}_{\varphi}(p_{\phi},p_{\phi^T})$ as long as there exists a convergent subsequence of $(\phi^k)_k$. The idea is the same as Theorem 3.2.4 in \citep{ChretienHeroProxGener}. Define $\phi^{k+1}$ by:
\[\phi^{k+1} = \arginf_{\phi} \hat{D}_{\varphi}(p_{\phi},p_{\phi^T}) + \beta_k D_{\psi}(\phi,\phi^k).\]
The proof of convergence is very simple and does not depend on the differentiability of any of the two functions $\hat{D}_{\varphi}$ or $D_{\psi}$. We only assume A1 and A2 to be verified. Let $(\phi^{N(k)})_k$ be a convergent subsequence. Let $\phi^{\infty}$ be its limit. This is guaranteed by the compactness of $\Phi^0$ and the fact that the whole sequence $(\phi^k)_k$ resides in $\Phi^0$ (see Proposition \ref{prop:DecreaseDphi}-b). Suppose also that the sequence $(\beta_k)_k$ converges to 0 as $k$ goes to infinity.\\ 
Let $\phi$ by a vector of $\Phi$ which has a value of $\hat{D}_{\varphi}(p_{\phi},p_{\phi^T})$ strictly inferior to the value of the same function at $\phi^{\infty}$, i.e.
\begin{equation}
\hat{D}_{\varphi}(p_{\phi},p_{\phi^T}) < \hat{D}_{\varphi}(p_{\phi^{\infty}},p_{\phi^T}).
\label{eqn:GlobalInfAlgo1}
\end{equation}
By Definition of $\phi^{N(k)}$, we have:
\[\hat{D}_{\varphi}(p_{\phi^{N(k)}},p_{\phi^T}) + \beta_{N(k)-1} D_{\psi}(\phi^{N(k)},\phi^{N(k)-1}) \leq \hat{D}_{\varphi}(p_{\phi},p_{\phi^T}) + \beta_{N(k)-1} D_{\psi}(\phi,\phi^{N(k)}),\]
which holds for every $\phi$ in the whole set $\Phi$. Using the non negativity of the term $\beta_{N(k)-1}D_{\psi}(\phi^{N(k)},\phi^{N(k)-1})$, one can write:
\begin{equation}
\hat{D}_{\varphi}(p_{\phi^{N(k)}},p_{\phi^T}) \leq \hat{D}_{\varphi}(p_{\phi},p_{\phi^T}) + \beta_{N(k)} D_{\psi}(\phi,\phi^{N(k)}).
\label{eqn:GlobalInfAlgo2}
\end{equation}
As we pass to the limit on $k$, recall firstly that $(\beta_k)_k$ converges to 0, so that any subsequence $(\beta_{N(k)})_k$ also converges to 0. Secondly, the continuity assumption on $D_{\psi}$ implies that, since $\phi^{N(k)}\rightarrow\phi^{\infty}$, $D_{\psi}(\phi,\phi^{N(k)})$ converges to $D_{\psi}(\phi,\phi^{\infty})$. By compactness of $\Phi^0$ and Proposition \ref{prop:DecreaseDphi}-b, we have $\phi^{\infty}\in\Phi^0$. The continuity again of $D_{\psi}$ will imply that the quantity $D_{\psi}(\phi,\phi^{\infty})$ is finite. Finally, inequality (\ref{eqn:GlobalInfAlgo2}) now implies that:
\[\hat{D}_{\varphi}(p_{\phi^{\infty}},p_{\phi^T}) \leq \hat{D}_{\varphi}(p_{\phi},p_{\phi^T})\]
which contradicts with the choice of $\phi$ verifying (\ref{eqn:GlobalInfAlgo1}). Hence, $\phi^{\infty}$ is a global infimum on $\Phi$.\\
The problem with this approach is that it depends heavily on the fact that the supremum on each step of the algorithm is calculated exactly. This does not happen in general unless function $\hat{D}_{\varphi}(p_{\phi},p_{\phi^T}) + \beta_k D_{\psi}(\phi,\phi^k)$ is convex or that we dispose of an algorithm which can solve perfectly non convex optimization problems\footnote{In this case, there is no meaning in applying an iterative proximal algorithm. We would have used the optimization algorithm directly on the objective function $\hat{D}_{\varphi}(p_{\phi},p_{\phi^T})$}. Although in our approach, we use a similar assumption to prove the consecutive decreasing of $\hat{D}_{\varphi}(p_{\phi},p_{\phi^T})$, we can replace the infimum calculus in (\ref{eqn:DivergenceAlgo}) by two things. We require that at each step we find a local infimum of $\hat{D}_{\varphi}(p_{\phi},p_{\phi^T}) + D_{\psi}(\phi,\phi^k)$ whose evaluation with $\phi\mapsto\hat{D}_{\varphi}(p_{\phi},p_{\phi^T})$ is less than the previous term of the sequence $\phi^k$. If we can no longer find any local minima verifying the claim, the procedure stops with $\phi^{k+1}=\phi^k$. This ensures the availability of all the proofs presented in this paper with no change.

\subsection{The two-component Gaussian mixture}\label{Example:GaussMix}
We suppose that the model $(p_{\phi})_{\phi\in\Phi}$ is a mixture of two gaussian densities, and that we are only interested in estimating the means $\mu=(\mu_1,\mu_2)\in\mathbb{R}^2$ and the proportion $\lambda\in[\eta,1-\eta]$. The use of $\eta$ is to avoid cancellation of any of the two components, and to keep the hypothesis $h_i(x|\phi)>0$ for $x=1,2$ verified. We also suppose that the components variances are reduced ($\sigma_i = 1$). The model takes the form 
\begin{equation}
p_{\lambda,\mu}(x) = \frac{\lambda}{\sqrt{2\pi}} e^{-\frac{1}{2}(x-\mu_1)^2} + \frac{1-\lambda}{\sqrt{2\pi}} e^{-\frac{1}{2}(x-\mu_2)^2}.
\label{eqn:GaussMixModel}
\end{equation} 
In the case of $\varphi(t)=-\log(t)+t-1$, the set $\Phi^0$ is given by $\Phi^0 = J^{-1}\left([J(\phi^0),+\infty)\right)$. The log-likelihood function $J$ is clearly of class $\mathcal{C}^1$(int($\Phi$)), where $\Phi = [\eta,1-\eta]\times\mathbb{R}^2$. The regularization term $D_{\psi}$ is defined by (\ref{eqn:DivergenceClassesNtNorm}) where:
\[h_i(1|\phi) = \frac{\lambda e^{-\frac{1}{2}(y_i-\mu_1)^2}}{\lambda e^{-\frac{1}{2}(y_i-\mu_1)^2} + (1-\lambda) e^{-\frac{1}{2}(y_i-\mu_2)^2}}, \quad h_i(2|\phi) = 1-h_i(1|\phi).\]
Functions $h_i$ are clearly of class $\mathcal{C}^1$(int($\Phi$)), hence, assumptions A1 and AC are verified. We prove that $\Phi^0$ is closed and bounded which is sufficient to conclude its compactness, since the space $[\eta,1-\eta]\times\mathbb{R}^2$ provided with the euclidean distance is complete.\\ 
\textbf{If we are using the dual estimator of the $\varphi-$divergence given by (\ref{eqn:DivergenceDef})}, then assumption A0 can be verified using the maximum theorem of \cite{Berge}. There is still a great difficulty in studying the properties (closedness or compactness) of the set $\Phi^0$. Moreover, all convergence properties of the sequence $\phi^k$ require the continuity of the estimated $\varphi-$divergence $\hat{D}_{\varphi}(p_{\phi},p_{\phi^T})$ with respect to $\phi$. In order to prove the continuity of the estimated divergence, we need to assume that $\Phi$ is compact, i.e. assume that the means are included in an interval of the form $[\mu_{\min},\mu_{\max}]$. Now, using Theorem 10.31 from \cite{Rockafellar}, $\phi\mapsto\hat{D}_{\varphi}(p_{\phi},p_{\phi^T})$ is continuous and differentiable almost everywhere with respect to $\phi$. \\
The compactness assumption of $\Phi$ implies directly the compactness of $\Phi^0$. Indeed
\begin{eqnarray*}
\Phi^0 & = & \left\{\phi\in\Phi, \hat{D}_{\varphi}(p_{\phi},p_{\phi^T})\leq \hat{D}_{\varphi}(p_{\phi^0},p_{\phi^T})\right\} \\
			 & = & \hat{D}_{\varphi}(p_{\phi},p_{\phi^T})^{-1}\left((-\infty, \hat{D}_{\varphi}(p_{\phi^0},p_{\phi^T})]\right).
\end{eqnarray*}
$\Phi^0$ is then the inverse image by a continuous function of a closed set, so it is closed in $\Phi$. Hence, it is compact.
\begin{conclusion}
\label{conc:conclusion1}
Using Propositions \ref{prop:NewRes} and \ref{prop:DecreaseDphi}, if $\Phi=[\eta,1-\eta]\times [\mu_{\min},\mu_{\max}]^2$, the sequence $(\hat{D}_{\varphi}(p_{\phi^k},p_{\phi^T}))_k$ defined through formula (\ref{eqn:DivergenceDef}) converges and there exists a subsequence $(\phi^{N(k)})$ which converges to a stationary point of the estimated divergence. Moreover, every limit point of the sequence $(\phi^k)_k$ is a stationary point of the estimated divergence. 
\end{conclusion}
\textbf{If we are using the kernel-based dual estimator given by (\ref{eqn:NewDivergenceDef})} with a Gaussian kernel density estimator, then function $\phi\mapsto \hat{D}_{\varphi}(p_{\phi},p_{\phi^T})$ is continuously differentiable over $\Phi$ even if the means $\mu_1$ and $\mu_2$ are not bounded. For example, take $\varphi = \varphi_{\gamma}$ defined by (\ref{eqn:CressieReadPhi}). There is one condition which relates the window of the kernel, say $w$, with the value of $\gamma$; $\gamma(w^2-1)>-1$. For $\gamma=2$ (the Pearson's $\chi^2$), we need that $w^2>1/2$. For $\gamma=1/2$ (the Hellinger), there is no condition on $w$.\\
Closedness of $\Phi^0$ is proved similarly to the previous case. Boundedness is however must be treated differently since $\Phi$ is not necessarily compact and is supposed to be $\Phi = [\eta,1-\eta]^s\times\mathbb{R}^s$. For simplicity take $\varphi=\varphi_{\gamma}$. The idea is to choose $\phi^0$ an initialization for the proximal algorithm in a way that $\Phi^0$ does not include unbounded values of the means. Continuity of $\phi\mapsto\hat{D}_{\varphi}(p_{\phi},p_{\phi^T})$ permits to calculate the limits when either (or both) of the means tends to infinity. If both means goes to infinity, then $p_{\phi}(x)\rightarrow 0,\forall x$. Thus, for $\gamma\in(0,\infty)\setminus \{1\}$, we have $\hat{D}_{\varphi}(p_{\phi},p_{\phi^T})\rightarrow \frac{1}{\gamma(\gamma-1)}$. For $\gamma<0$, the limit is infinity. If only one of the means tends to $\infty$, then the corresponding component vanishes from the mixture. Thus, if we choose $\phi^0$ such that:
\begin{eqnarray}
\hat{D}_{\varphi}(p_{\phi^0},p_{\phi^T}) & < & \min\left(\frac{1}{\gamma(\gamma-1)}, \inf_{\lambda,\mu}\hat{D}_{\varphi}(p_{(\lambda,\infty,\mu)},p_{\phi^T})\right) \text{ if } \gamma \in(0,\infty)\setminus \{1\};\label{eqn:CondGaussMixNewDual1} \\
\hat{D}_{\varphi}(p_{\phi^0},p_{\phi^T}) & < & \inf_{\lambda,\mu}\hat{D}_{\varphi}(p_{(\lambda,\infty,\mu)},p_{\phi^T})\qquad \text{ if } \gamma <0,
\label{eqn:CondGaussMixNewDual2}
\end{eqnarray}
then the algorithm starts at a point of $\Phi$ whose function value is inferior to the limits of $\hat{D}_{\varphi}(p_{\phi},p_{\phi^T})$ at infinity. By Proposition \ref{prop:DecreaseDphi}, the algorithm will continue to decrease the value of $\hat{D}_{\varphi}(p_{\phi},p_{\phi^T})$ and never goes back to the limits at infinity. Besides, the Definition of $\Phi^0$ permits to conclude that if $\phi^0$ is chosen according to condition (\ref{eqn:CondGaussMixNewDual1},\ref{eqn:CondGaussMixNewDual2}), then $\Phi^0$ is bounded. Thus, $\Phi^0$ becomes compact. Unfortunately the value of $\inf_{\lambda,\mu}\hat{D}_{\varphi}(p_{(\lambda,\infty,\mu)},p_{\phi^T})$ can be calculated but numerically. We will see next that in the case of Likelihood function, a similar condition will be imposed for the compactness of $\Phi^0$, and there will be no need for any numerical calculus.
\begin{conclusion}
\label{conc:conclusion2}
Using Propositions \ref{prop:NewRes} and \ref{prop:DecreaseDphi}, under condition (\ref{eqn:CondGaussMixNewDual1}, \ref{eqn:CondGaussMixNewDual2}) the sequence $(\hat{D}_{\varphi}(p_{\phi^k},p_{\phi^T}))_k$ defined through formula (\ref{eqn:NewDivergenceDef}) converges and there exists a subsequence $(\phi^{N(k)})$ which converges to a stationary point of the estimated divergence. Moreover, every limit point of the sequence $(\phi^k)_k$ is a stationary point of the estimated divergence. 
\end{conclusion}
\textbf{In the case of the likelihood $\varphi(t)=-\log(t)+t-1$}, the set $\Phi^0$ can be written as:
\begin{eqnarray*}
\Phi^0 & = & \left\{\phi\in\Phi, J(\phi)\geq J(\phi^0)\right\} \\
			 & = & J^{-1}\left([J(\phi^0),+\infty)\right),
\end{eqnarray*}
where $J$ is the log-likelihood function. Function $J$ is clearly of class $\mathcal{C}^1$ on (int($\Phi$)). We prove that $\Phi^0$ is closed and bounded which is sufficient to conclude its compactness, since the space $[\eta,1-\eta]^s\times\mathbb{R}^s$ provided with the euclidean distance is complete.\\ 
\textbf{Closedness.} The set $\Phi^0$ is the inverse image by a continuous function (the log-likelihood). Therefore it is closed in $[\eta,1-\eta]^s\times\mathbb{R}^s$.\\
\textbf{Boundedness.} By contradiction, suppose that $\Phi^0$ is unbounded, then there exists a sequence $(\phi^l)_l$ which tends to infinity. Since $\lambda^l\in[\eta,1-\eta]$, then either of $\mu_1^l$ or $\mu_2^l$ tends to infinity. Suppose that both $\mu_1^l$ and $\mu_2^l$ tend to infinity, we then have $J(\phi^l) \rightarrow -\infty$. Any finite initialization $\phi^0$ will imply that $J(\phi^0)>-\infty$ so that $\forall \phi\in\Phi^0, J(\phi)\geq J(\phi^0)>-\infty$. Thus, it is impossible for both $\mu_1^l$ and $\mu_2^l$ to go to infinity.\\
Suppose that $\mu_1^l \rightarrow \infty$, and that $\mu_2^l$ converges\footnote{Normally, $\mu_2^l$ is bounded; still, we can extract a subsequence which converges.} to $\mu2$. The limit of the likelihood has the form:
\[L(\lambda, \infty, \phi_2) = \prod_{i=1}^n{\frac{(1-\lambda)}{\sqrt{2\pi}}e^{-\frac{1}{2}(y_i-\mu_2)^2}},\]
which is bounded by its value for $\lambda = 0$ and $\mu_2 = \frac{1}{n}\sum_{i=1}^n{y_i}$. Indeed, since $1-\lambda\leq 1$, we have:
\[L(\lambda, \infty, \phi_2) \leq \prod_{i=1}^n{\frac{1}{\sqrt{2\pi}}e^{-\frac{1}{2}(y_i-\mu_2)^2}}.\]
The right hand side of this inequality is the likelihood of a Gaussian model $\mathcal{N}(\mu_2,0)$, so that it is maximized when $\mu_2=\frac{1}{n}\sum_{i=1}^n{y_i}$. Thus, if $\phi^0$ is chosen in a way that $J(\phi^0)>J\left(0,\infty,\frac{1}{n}\sum_{i=1}^n{y_i}\right)$, the case when $\mu_1$ tends to infinity and $\mu_2$ is bounded would never be allowed. For the other case where $\mu_2\rightarrow\infty$ and $\mu_1$ is bounded, we choose $\phi^0$ in a way that $J(\phi^0)>J\left(1,\frac{1}{n}\sum_{i=1}^n{y_i},\infty\right)$. In conclusion, with a choice of $\phi^0$ such that:
\begin{equation}
J(\phi^0)>\max\left[J\left(0,\infty,\frac{1}{n}\sum_{i=1}^n{y_i}\right),\; J\left(1,\frac{1}{n}\sum_{i=1}^n{y_i},\infty\right)\right]
\label{eqn:TwoGaussMixCond}
\end{equation}
the set $\Phi^0$ is bounded.\\
This condition on $\phi^0$ is very natural and means that we need to begin at a point at least better than the extreme cases where we only have one component in the mixture. This can be easily verified by choosing a random vector $\phi^0$, and calculate the corresponding log-likelihood value. If $J(\phi^0)$ does not verify the previous condition, we draw again another random vector until satisfaction.
\begin{conclusion}
\label{conc:conclusion3}
Using Propositions \ref{prop:NewRes} and \ref{prop:DecreaseDphi}, under condition (\ref{eqn:TwoGaussMixCond}) the sequence $(J(\phi^k))_k$ converges and there exists a subsequence $(\phi^{N(k)})$ which converges to a stationary point of the likelihood function. Moreover, every limit point of the sequence $(\phi^k)_k$ is a stationary point of the likelihood. 
\end{conclusion}
\textbf{Assumption A3 is not fulfilled} (this part applies for all aforementioned situations). As mentioned in the paper of \citep{Tseng}, for the two Gaussian mixture Example, by changing $\mu_1$ and $\mu_2$ by the same amount and suitably adjusting $\lambda$, the value of $h_i(x|\phi)$ would be unchanged. We explore this more thoroughly by writing the corresponding equations. Let's suppose, by absurd, that for distinct $\phi$ and $\phi'$ we have $D_{\psi}(\phi|\phi')=0$. By Definition of $D_{\psi}$, it is given by a sum of non negative terms, which implies that all terms need to be equal to zero. The following lines are equivalent $\forall i \in \{1,\cdots,n\}$:
\begin{eqnarray*}
h_i(0|\lambda,\mu_1,\mu_2) & = & h_i(0|\lambda',\mu'_1,\mu'_2) \\
\frac{\lambda e^{-\frac{1}{2}(y_i-\mu_1)^2}}{\lambda e^{-\frac{1}{2}(y_i-\mu_1)^2} + (1-\lambda) e^{-\frac{1}{2}(y_i-\mu_2)^2}} & = & \frac{\lambda' e^{-\frac{1}{2}(y_i-\mu'_1)^2}}{\lambda' e^{-\frac{1}{2}(y_i-\mu'_1)^2} + (1-\lambda') e^{-\frac{1}{2}(y_i-\mu'_2)^2}} \\
\log\left(\frac{1-\lambda}{\lambda}\right) - \frac{1}{2}(y_i-\mu_2)^2 + \frac{1}{2}(y_i-\mu_1)^2 & = & \log\left(\frac{1-\lambda'}{\lambda'}\right) - \frac{1}{2}(y_i-\mu'_2)^2 + \frac{1}{2}(y_i-\mu'_1)^2
\end{eqnarray*}
Looking at this set of $n$ equations as an equality of two polynomials on $y$ of degree 1 at $n$ points, we deduce that as we dispose of two distinct observations, say, $y_1$ and $y_2$, the two polynomials need to have the same coefficients. Thus the set of $n$ equations is equivalent to the following two equations:
\begin{equation}
\left\{\begin{array}{ccc}\mu_1-\mu_2 & = & \mu'_1-\mu'_2 \\
		\log\left(\frac{1-\lambda}{\lambda}\right) + \frac{1}{2}\mu_1^2 - \frac{1}{2}\mu_2^2 & = & \log\left(\frac{1-\lambda'}{\lambda'}\right) + \frac{1}{2}{\mu'_1}^2 - \frac{1}{2}{\mu'_2}^2	
		\end{array}\right.
\label{eqn:EqSysGaussMix}
\end{equation}
These two equations with three variables have an infinite number of solutions. Take for example $\mu_1 = 0,\mu_2=1,\lambda=\frac{2}{3},\mu'_1=\frac{1}{2}, \mu'_2=\frac{3}{2},\lambda'=\frac{1}{2}$. \\
\begin{remark} 
The previous conclusion can be extended to any two-component mixture of exponential families having the form:
\[p_{\phi}(y) = \lambda e^{\sum_{i=1}^{m_1}{\theta_{1,i}y^{i}} - F(\theta_1)} + (1-\lambda)e^{\sum_{i=1}^{m_2}{\theta_{2,i}y^{i}} - F(\theta_2)}.\]
One may write the corresponding $n$ equations. The polynomial of $y_i$ has a degree of at most $\max(m_1,m_2)$. Thus, if one disposes of $\max(m_1,m_2)+1$ distinct observations, the two polynomials will have the same set of coefficients. Finally, if $(\theta_1,\theta_2)\in\mathbb{R}^{d-1}$ with $d>\max(m_1,m_2)$, then assumption A3 does not hold.
\end{remark}
Unfortunately, we have no an information about the difference between consecutive terms $\|\phi^{k+1}-\phi^k\|$ except for the case of $\psi(t) = \varphi(t)=-\log(t)+t-1$ which corresponds to the classical EM recurrence:
\[\lambda^{k+1} = \frac{1}{n}\sum_{i=1}^n{h_i(0|\phi^k)},\quad \mu_1^{k+1} = \frac{\sum_{i=1}^n{y_ih_i(0|\phi^k)}}{\sum_{i=1}^n{h_i(0|\phi^k)}}\quad \mu_1^{k+1} = \frac{\sum_{i=1}^n{y_ih_i(1|\phi^k)}}{\sum_{i=1}^n{h_i(1|\phi^k)}}.\]
Tseng \cite{Tseng} has shown that we can prove directly that $\phi^{k+1}-\phi^k$ converges to 0.
\section{Simulation Study}\label{sec:Simulations}
We summarize the results of 100 experiments on $100$-samples by giving the average of the estimates and the error committed, and the corresponding standard deviation. The criterion error is the total variation distance (TVD) which is calculated using the $L1$ distance. Indeed, the Scheff\'e lemma (see \cite{Meister} page 129.) states that:
\[\sup_{A\in\mathcal{B}_n(\mathbb{R})}\left|dP_{\phi}(A) - dP_{\phi^T}(A)\right| = \frac{1}{2}\int_{\mathbb{R}}{\left|p_{\phi}(y) - p_{\phi^T}(y)\right|dy}.\]
The TVD gives a measure of the maximum error we may commit when we use the estimated model in lieu of the true distribution. We consider the Hellinger divergence for estimators based on $\varphi-$divergences which corresponds to $\varphi(t)=\frac{1}{2}(\sqrt{t}-1)^2$. Our preference of the Hellinger divergence is that we hope to obtain robust estimators without loss of efficiency, see \cite{Jimenez}. $D_{\psi}$ is calculated with $\psi(t)=\frac{1}{2}(\sqrt{t}-1)^2$. The kernel-based MD$\varphi$DE is calculated using the Gaussian kernel, and the window is calculated using Silverman's rule. We included in the comparison the minimum density power divergence (MDPD) of \cite{BasuMPD}. The estimator is defined by:
\begin{eqnarray}
\hat{\phi}_n & = & \arginf_{\phi\in\Phi} \int{p_{\phi}^{1+a}}(z) dz - \frac{a+1}{a}\frac{1}{n}\sum_{i}^n{p_{\phi}^{a}(y_i)} \nonumber \\
 & = & \arginf_{\phi\in\Phi} \mathbb{E}_{P_{\phi}}\left[p_{\phi}^a\right] - \frac{a+1}{a}\mathbb{E}_{P_n}\left[p_{\phi}^{a}\right].
\label{eqn:MDPDdef}
\end{eqnarray}
where $a\in(0,1]$. This is a Bregman divergence and is known to have good efficiency and robustness for a good choice of the tradeoff parameter. According to the simulation results in \cite{Diaa}, the value of $a=0.5$ seems to give a good tradeoff between robustness against outliers and a good performance under the model. Notice that the MDPD coincides with MLE when $a$ tends to zero. Thus, our methodology presented here in this article, is applicable on this estimator and the proximal-point algorithm can be used to calculate the MDPD. The proximal term will be kept the same, i.e $\psi(t)=\frac{1}{2}(\sqrt{t}-1)^2$.\\
Simulations from two mixture models are given below; a Gaussian mixture and a Weibull mixture. The MLE for both mixtures was calculated using the EM algorithm. \\
Optimizations were carried out using the Nelder-Mead algorithm \cite{NelderMead} under the statistical tool R \cite{Rtool}. Numerical integrations in the Gaussian mixture were calculated using the \texttt{distrExIntegrate} function of package \texttt{distrEx}. It is a slight modification of the standard function \texttt{integrate}. It performs a Gauss-Legendre quadrature when function \texttt{integrate} returns an error. In the Weibull mixture, we used the \texttt{integral} function from package \texttt{pracma}\footnote{Function \texttt{integral} includes a variety of adaptive numerical integration methods such as Kronrod-Gauss quadrature, Romberg's method, Gauss-Richardson quadrature, Clenshaw-Curtis (not adaptive) and (adaptive) Simpson's method.}. Although being slow, it performs better than other functions even if the integrand has a relatively bad behavior.


\subsection{The two-component Gaussian mixture revisited}\label{Example:DivergenceMixture}
We consider the Gaussian mixture (\ref{eqn:GaussMixModel}) presented earlier with true parameters $\lambda = 0.35, \mu_1 = -2, \mu_2 = 1.5$ and known variances equal to 1. Contamination was done by adding in the original sample to the 5 lowest values random observations from the uniform distribution $\mathcal{U}[-5,-2]$. We also added to the 5 largest values random observations from the uniform distribution $\mathcal{U}[2,5]$. Results are summarized in Table (\ref{tab:ErrGaussMix}). The EM algorithm was initialized according to condition (\ref{eqn:TwoGaussMixCond}). This condition gave good results when we are under the model whereas it did not result always in good estimates (the proportion converged towards 0 or 1) when outliers were added, and thus the EM algorithm was reinitialized manually. \\	
Figure (\ref{fig:DecreaseDivGaussChi2Chi2}) shows the values of the estimated divergence for both formulas (\ref{eqn:DivergenceDef}) and (\ref{eqn:NewDivergenceDef}) on a logarithmic scale at each iteration of the algorithm.
\begin{figure}[ht]
\centering
\includegraphics[scale = 0.5]{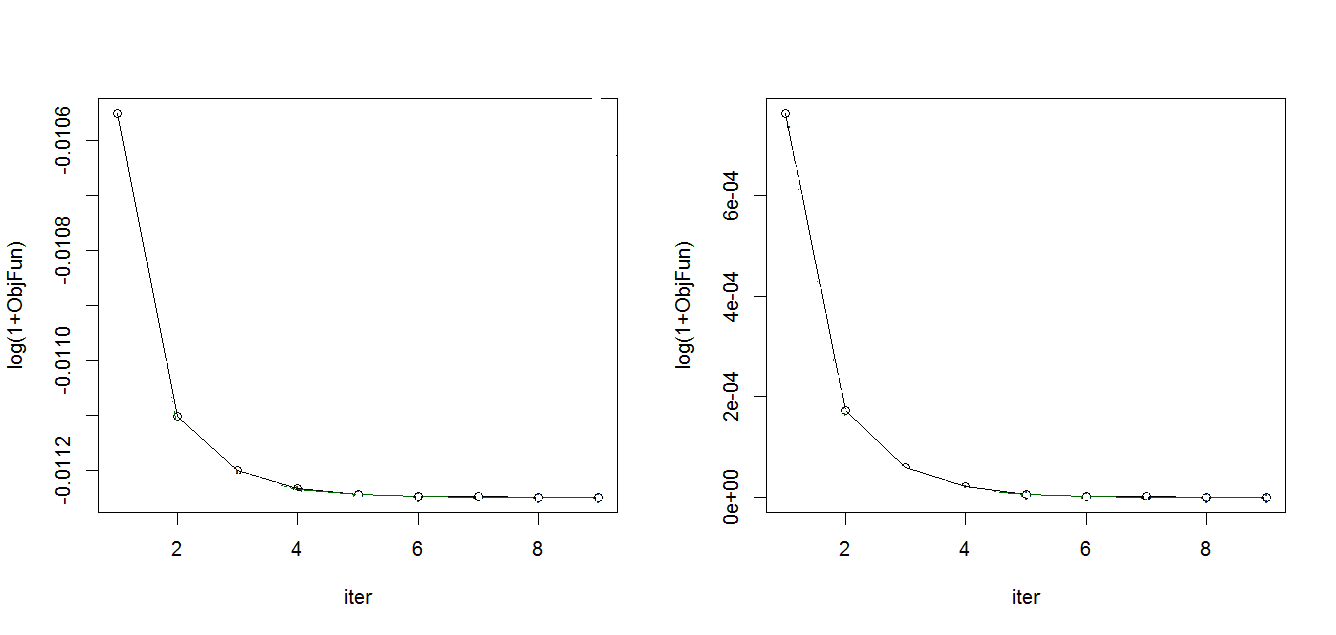}
\caption{Decrease of the (estimated) Hellinger divergence between the true density and the estimated model at each iteration in the Gaussian mixture. The figure to the left is the curve of the values of the kernel-based dual formula (\ref{eqn:NewDivergenceDef}). The figure to the right is the curve of values of the classical dual formula (\ref{eqn:DivergenceDef}). Values are taken at a logarithmic scale $\log(1+x)$.}
\label{fig:DecreaseDivGaussChi2Chi2}
\end{figure}

In what concerns our simulation results. The total variation of all four estimation methods is very close when we are under the model. When we added outliers, the classical MD$\varphi$DE was as sensitive as the maximum likelihood estimator. The error was doubled. Both the kernel-based MD$\varphi$DE and the MDPD are clearly robust since the total variation of these estimators under contamination has slightly increased.
\begin{table}[ht]
\centering
\caption{The mean and the standard deviation of the estimates and the errors committed in a 100-run experiment of a two-component gaussian mixture. The true set of parameters is $\lambda = 0.35, \mu_1 = -2, \mu_2 = 1.5$.}
\begin{tabular}{|l|c|c|c|c|c|c||c|c|}
\hline
 Estimation method & $\lambda$ & sd($\lambda$) & $\mu_1$ & sd($\mu_1$) & $\mu_2$ & sd($\mu_2$) & TVD & sd(TVD)\\
 \hline
 \hline
\multicolumn{9}{|c|}{Without Outliers} \\
\hline
\hline
Classical MD$\varphi$DE & 0.349 & 0.049 & -1.989 & 0.207 & 1.511 & 0.151 & 0.061 & 0.029\\
New MD$\varphi$DE - Silverman & 0.349 & 0.049 & -1.987 & 0.208 & 1.520 & 0.155 & 0.062 & 0.029\\
MDPD $a=0.5$ & 0.360 & 0.053 & -1.997 & 0.226 & 1.489 & 0.135 & 0.065 & 0.025 \\
EM (MLE) & 0.360 & 0.054 & -1.989 & 0.204 & 1.493 & 0.136 & 0.064 & 0.025\\
\hline
\hline
\multicolumn{9}{|c|}{With $10\%$ Outliers} \\
\hline
\hline
Classical MD$\varphi$DE & 0.357 & 0.022 & -2.629 & 0.094 & 1.734 & 0.111 & 0.146 & 0.034\\
New MD$\varphi$DE - Silverman & 0.352 & 0.057 & -1.756 & 0.224 & 1.358 & 0.132 & 0.087 & 0.033\\
MDPD $a=0.5$ & 0.364 & 0.056 & -1.819 & 0.218 & 1.404 & 0.132 & 0.078 & 0.030 \\
EM (MLE) & 0.342 & 0.064 & -2.617 & 0.288 & 1.713 & 0.172 & 0.150 & 0.034\\
\hline
\end{tabular}
\label{tab:ErrGaussMix}
\end{table}


\subsection{The two-component Weibull mixture model}
We consider a two-component Weibull mixture with unknown shapes $\nu_1 = 1.2, \nu_2 = 2$ and a proportion $\lambda = 0.35$. The scales are known an equal to $\sigma_1=0.5, \sigma_2=2$. The desity function is given by:
\begin{equation}
p_{\phi}(x) = 2\lambda\alpha_1 (2x)^{\alpha_1-1} e^{-(2x)^{\alpha_1}}+(1-\lambda)\frac{\alpha_2}{2}\left(\frac{x}{2}\right)^{\alpha_2-1} e^{-\left(\frac{x}{2}\right)^{\alpha_2}}.
\label{eqn:WeibullMixture}
\end{equation}
Contamination was done by replacing 10 observations of each sample chosen randomly by 10 i.i.d. observations drawn from a Weibull distribution with shape $\nu = 0.9$ and scale $\sigma = 3$. Results are summarized in Table (\ref{tab:ErrWeibullMix}). Notice that it would have been better to use asymmetric kernels in order to build the kernel-based MD$\varphi$DE since their use in the context of positive-supported distributions is advised in order to reduce the bias at zero, see \cite{Diaa} for a detailed comparison with symmetric kernels. This is not however the goal of this paper, and besides, the use of symmetric kernels in this mixture model gave satisfactory results.\\
Simulations results in table \ref{tab:ErrWeibullMix} confirms once more the validity of our proximal-point algorithm and the clear robustness of both the kernel-based MD$\varphi$DE and the MDPD.
\begin{table}[ht]
\centering
\caption{The mean and the standard deviation of the estimates and the errors committed in a 100-run experiment of a two-component Weibull mixture. The true set of parameter is $\lambda = 0.35, \nu_1 = 1.2, \nu_2 = 2$.}
\begin{tabular}{|l|c|c|c|c|c|c||c|c|}
\hline
 Estimation method & $\lambda$ & sd($\lambda$) & $\nu_1$ & sd($\nu_1$) & $\nu_2$ & sd($\nu_2$) & TVD & sd(TVD)\\
 \hline
 \hline
\multicolumn{9}{|c|}{Without Outliers} \\
\hline
\hline
Classical MD$\varphi$DE & 0.356 & 0.066 & 1.245 & 0.228 & 2.055 & 0.237 & 0.052 & 0.025\\
New MD$\varphi$DE - Silverman & 0.387 & 0.067 & 1.229 & 0.241 & 2.145 & 0.289 & 0.058 & 0.029\\
MDPD $a=0.5$ & 0.354 & 0.068 & 1.238 & 0.230 & 2.071 & 0.345 & 0.056 & 0.029 \\
EM (MLE) & 0.355 & 0.066 & 1.245 & 0.228 & 2.054 & 0.237 & 0.052 & 0.025\\
\hline
\hline
\multicolumn{9}{|c|}{With $10\%$ Outliers} \\
\hline
\hline
Classical MD$\varphi$DE & 0.250 & 0.085 & 1.089 & 0.300 & 1.470 & 0.335 & 0.092 & 0.037\\
New MD$\varphi$DE - Silverman &  0.349 & 0.076 & 1.122 & 0.252 & 1.824 & 0.324 & 0.067 & 0.034\\
MDPD $a=0.5$ & 0.322 & 0.077 & 1.158 & 0.236 & 1.858 & 0.344 & 0.060 & 0.029 \\
EM (MLE) & 0.259  & 0.095 & 0.941 & 0.368 & 1.565 & 0.325 & 0.095 & 0.035\\
\hline
\end{tabular}
\label{tab:ErrWeibullMix}
\end{table}

\section{Conclusions and perspectives}
We introduced in this paper a proximal-point algorithm which permit to calculate divergence-based estimators. We studied the theoretical convergence of the algorithm and verified it on a two-component Gaussian mixture. We performed several simulations which confirmed that the algorithm works and is a way to calculate divergence-based estimators. We also applied our proximal algorithm on a Bregman divergence estimator (the MDPD), and the algorithm succeeded to produce the MDPD. Further investigations about the role of the proximal term are needed and may be considered in a future work.

\bibliographystyle{plainnat}
\bibliography{Biblio}
\end{document}